\documentclass[11pt]{article}
\usepackage{amsmath}
\usepackage{amsfonts}
\usepackage{amssymb}
\usepackage{latexsym}
\usepackage[dvips]{graphicx}
\newtheorem{theorem}{\bf Theorem}[section]
\newtheorem{lemma}[theorem]{\bf Lemma}
\newtheorem{prop}[theorem]{\bf Proposition}

\newenvironment{proof}{\noindent{\em Proof:}}{\quad \hfill$\Box$\vspace{2ex}}

\def \bN {\Bbb N}
\def \bZ {\Bbb Z}

\def \bR {\Bbb R}

\def \bH {\Bbb H}

\def \bU {\Bbb U}

\def \bC {\Bbb C}

\def \bS {\Bbb S}

\def \ba {{\bf a}}

\def \cA {{\cal A}}
\def \cB {{\cal B}}

\def \cH {{\cal H}}

\def \cM {{\cal M}}
\def \cS {{\cal S}}
\def \cQ {{\cal Q}}

\def \ba {{\bf a}}
\def \and {\, \mbox{\rm and}\, }

\def \co {\,{\rm conv}\,}
\def \span {\,{\rm span}\,}
\def \supp {\,{\rm supp}\,}
\def \sgn {\,{\rm sgn}\,}

\def \Re {\,{\rm Re}\,}

\def \ae {\mbox{a.e. }}
\makeatletter

\newcommand{\Rmnum}[1]{\expandafter\@slowromancap\romannumeral #1@}
\makeatother
\begin{document}
\title{\bf Multidimensional Analytic Signals and\\ the Bedrosian Identity\thanks{The Project was partially supported by Guangdong Provincial Government of China through the
``Computational Science Innovative Research Team" program.}}
\author{Haizhang Zhang\thanks{School of Mathematics
and Computational Science and Guangdong Province Key Laboratory of Computational Science,
 Sun Yat-sen University, Guangzhou 510275, P. R. China. E-mail address: {\it zhhaizh2@sysu.edu.cn}. Supported in part by Natural Science Foundation of China under grants 11222103, 11101438 and 91130009, and by the US Army Research Office.}}
\date{}

\maketitle
\begin{abstract} The analytic signal method via the Hilbert transform is a key tool in signal analysis and processing, especially in the time-frquency analysis. Imaging and other applications to multidimensional signals call for extension of the method to higher dimensions. We justify the usage of partial Hilbert transforms to define multidimensional analytic signals from both engineering and mathematical perspectives. The important associated Bedrosian identity $T(fg)=fTg$ for partial Hilbert transforms $T$ are then studied. Characterizations and several necessity theorems are established. We also make use of the identity to construct basis functions for the time-frequency analysis.


\end{abstract}
\vspace*{0.3cm}
\section{Introduction}
\setcounter{equation}{0}

Signals are carrier of information. In many applications, feature
information of a signal usually persists in the time-frequency
domain \cite{C}. The analytic signal method \cite{Gabor} is a
classical way of defining without ambiguity the local amplitude and
frequency of one-dimensional signals. Among others, it has been
proven useful in meteorological and atmospheric applications, ocean
engineering, structural science, and imaging processing,
\cite{HA,HS,HZL,OL,P}. Especially, it motivates the widely used
empirical mode decomposition and the Hilbert-Huang transform
\cite{CHRX,HA,Hu,HS}.

The analytic signal method makes use of the Hilbert transform $H$
defined for functions $f\in L^2(\bR)$ as
\begin{equation}\label{hilbert}
(H
f)(x):=\mbox{p.v.}\frac{1}{\pi}\int_{\bR}\frac{f(y)}{x-y}dy:=\lim_{{\varepsilon\rightarrow0^+}}\frac{1}{\pi}\int_{
|y-x|\ge\varepsilon}\frac{f(y)}{x-y}dy,\ \   x\in\bR.
\end{equation}
Let $\bH^1(\bR)$ be the Banach space of all the functions $\phi\in
L^1(\bR)$ such that $H \phi\in L^1(\bR)$ equipped with the norm
$$
\|\phi\|_{\bH^1(\bR)}:=\|\phi\|_{L^1(\bR)}+\|H \phi\|_{L^1(\bR)}.
$$
The Hilbert transform $H f$ of $f\in L^\infty(\bR)$ is defined as a
BMO function that is in the dual of $\bH^1(\bR)$ by
$$
\int_{\bR}(H f) \phi dx=-\int_{\bR}fH\phi dx,\ \ \phi\in
\bH^1(\bR).
$$
The {\it analytic signal} $\cA f$ of a real $f\in L^2(\bR)$ is
formed by adding to it its Hilbert transform as the imaginary part:
\begin{equation}\label{analytic}
\cA f:=f+iH f.
\end{equation}
With the amplitude-phase decomposition
$$
(\cA f)(t)=\rho(t)e^{i\theta(t)},\ \ t\in\bR,
$$
the $\rho(t)$ and $\theta'(t)$ respectively are taken as the
instantaneous amplitude and frequency of the signal $f$ at time $t$.

The Hilbert transform in (\ref{analytic}) has many favorable
mathematical properties that account for the usefulness of the
analytic signal method in the time-frequency analysis. Here we
mention four of them. Let $\bN$ be the set of positive integers and
$\bZ_+:=\bN\cup\{0\}$. For a fixed $d\in\bN$, we denote by
$\|\cdot\|$ and $(\cdot,\cdot)$ the standard Euclidean norm and
inner product on $\bR^d$. The {\it Schwartz class} $\cS(\bR^d)$
consists of infinitely differentiable functions $\varphi$ on $\bR^d$
such that for all $p\in\bZ_+$ and $q\in\bZ_+^d$
$$
\sup\{(1+\|x\|^2)^p|f^{(q)}(x)|:x\in\bR^d\}<+\infty.
$$
The Fourier transform $\hat{\varphi}$ of $\varphi\in\cS(\bR^d)$ is
defined as
$$
\hat{\varphi}(\xi):=\int_{\bR^d}\varphi(x)e^{-i(x,\xi)}dx,\ \
\xi\in\bR^d,
$$
which is again a function in $\cS(\bR^d)$. The Fourier transform can
be extended to the space $\cS'(\bR^d)$ of temperate distributions on
$\bR^d$ by a duality principle \cite{GW}.

The Hilbert transform has an equivalent definition via the Fourier
multiplier $-i\sgn$, where $\sgn(\xi)$ takes value $-1,0,1$ for
$\xi<0$, $\xi=0$ and $\xi>0$, respectively. Specifically, we have
for all $f\in L^2(\bR)$
\begin{equation}\label{hilbertsgn}
(H f)\hat{\,}(\xi)=-i\sgn(\xi)f(\xi),\ \  \xi\in\bR.
\end{equation}
There are two important consequences of this. First of all, we have
for all $f\in L^2(\bR)$ that
\begin{equation}\label{positive}
\supp(f+iH f)\subseteq\bR_+
\end{equation}
where $\bR_+:=[0,+\infty)$. In other words, the analytic signal
suppresses all the negative frequency components of the original
signal. This justifies the physical soundness of the analytic signal
method. Secondly, the multiplier definition (\ref{hilbertsgn}) makes
it possible to develop fast algorithms making use of the lighting
fast FFT to compute the analytic signal \cite{Ma}. In addition to
the above two properties, it was shown in \cite{Vakman1,Vakman2}
that the Hilbert transform is the only continuous and homogeneous
operator $L$ such that
$$
L(\cos(\omega_0t+\phi_0))=\sin(\omega_0t+\phi_0),\ \ \mbox{ for all
}\omega_0>0,\ \phi_0\in\bR.
$$
This implies that the analytic signal method is the only one that
satisfies some reasonable physical requirements
\cite{Vakman1,Vakman2}. Finally, the Hilbert transform satisfies the
Bedrosian theorem \cite{Bed}: If $f,g\in L^2(\bR)$ satisfy either
$\supp\hat{f}\subseteq\bR_+$, $\supp\hat{g}\subseteq\bR_+$ or
$\supp\hat{f}\subseteq[-a,a]$,
$\supp\hat{g}\subseteq(-\infty,-a]\cup[a,\infty)$ for some positive
number $a$, then there holds the {\it Bedrosian identity}
\begin{equation}\label{bedrosian1}
[H(fg)](x)=f(x)(H g)(x),\ \   x\in\bR.
\end{equation}
The Bedrosian identity simplifies the calculation of the Hilbert
transform of a product of functions, helps understand the
instantaneous amplitude and frequency of signals, and provides a
method of constructing basic signals in the time-frequency analysis
\cite{Bed,Br2,Br,C,NB,P,QWXZ,Stark,YZ2}. It has attracted much
interest from the mathematical community
\cite{Cerejeiras,ChenMicchelli,QCL,QXYYY,VZ,WangSilei,XY,XZ,YZ1,YZ2}. Here, we mention an observation
in \cite{VZ}. It states that the Hilbert transform is essentially
the only operator that satisfies the Bedrosian identity.

We conclude that the analytic signal (\ref{analytic}) is justified by the
above mathematical properties. This paper is motivated by the need
\cite{BS,Ha,XLLR} of defining multidimensional analytic
signals for the time-frequency analysis of multidimensional signals.
Naturally, we are inclined to define the analytic signal of $f\in
L^2(\bR^d)$ through a fixed operator $T:L^2(\bR^d)\to L^2(\bR^d)$ as
$f+Tf$. Thus the definition reduces to the choice to the operator
$T$. In references \cite{BS,Ha}, compositions of the partial
Hilbert transforms are used. One of our purposes is to show that linear
combinations of compositions of the partial Hilbert transforms are
indeed the only choice from the viewpoint of similar mathematical
properties introduced in the last paragraph. The justifications will
be carried out in the next section. In Section 3, we investigate the
multidimensional Bedrosian identity $T(fg)=fTg$ for $f,g\in
L^2(\bR^d)$, where $T$ is a linear combination of compositions of
the partial Hilbert transforms. In particular, a necessary and
sufficient condition and the Bedrosian theorem will be established.
The necessity of the Bedrosian theorem will be discussed as well. In
Section 4, we construct basis multidimensional analytic signals by
the results on the Bedrosian identity.

\section{Multidimensional Analytic Signals}
\setcounter{equation}{0}

Set $\bZ_n:=\{0,1,\ldots,n-1\}$ and $\bN_n:=\{1,2,\ldots,n\}$,
$n\in\bN$. The partial Hilbert transform $H_j$, $j\in\bN_d$, is
defined for $f\in L^2(\bR^d)$ as
$$
(H_jf)(x):=\mbox{p.v.}\frac{1}{\pi}\int_{\bR^d}\frac{f(y)}{x_j-y_j}dy:=\lim_{{\varepsilon\rightarrow0^+}}\frac{1}{\pi}\int_{
|y_j-x_j|\ge\varepsilon}\frac{f(y)}{x_j-y_j}dy,\ \   x\in\bR^d.
$$
It is known that $H_j$ is a bounded linear operator from
$L^2(\bR^d)$ to $L^2(\bR^d)$, \cite{Stein}. Moreover, there holds
\begin{equation}\label{multiplier2}
(H_jf)\hat{\,}(\xi)=-i\sgn(\xi_j)f(\xi),\ \   \xi\in\bR^d,\ \
j\in\bN_d.
\end{equation}
The purpose of the this section is to justify the using of linear
combinations of compositions of the partial Hilbert transforms to
define the multidimensional analytic signal. For simplicity, we
denote by $H_0$ the identity operator on $L^2(\bR^d)$. A linear
combination $T$ of compositions of the partial Hilbert transforms
has the form
\begin{equation}\label{composition}
T=\sum_{\alpha\in\bZ_{d+1}^d}c_\alpha\cH_\alpha,
\end{equation}
where $\cH_\alpha:=\prod_{j\in\bN_d}H_{\alpha_j}$ and $c_\alpha$ are
complex constants. We shall argue from four points of view that the
multidimensional analytic signal $\cA f$ of $f\in L^2(\bR^d)$ should
be defined as
\begin{equation}\label{analytic2}
\cA f:= f+Tf
\end{equation}
via an operator $T$ of the form (\ref{composition}).

\subsection{Concentration of the Frequency} By (\ref{positive}), the negative Fourier frequency of a one-dimensional analytic signal is
suppressed while the energy of the positive frequency is doubled.
For some applications, it is desirable to restrict the frequency to
a certain $d$-hyperoctant, especial the first one, \cite{Ha}.
Following the notations of \cite{VZ}, we let $\nu^k$,
$k\in\bN_{2^d}$, be the extreme points of the cube $[-1,1]^d$.
Assume that $\nu^1$ and $\nu^2$ are the point each of whose
component is $1$ and $-1$, respectively. The Euclidean space $\bR^d$
is divided into $2^d$ $d$-hyperoctants as
$$
\bR^d=\bigcup_{k\in\bN_{2^d}}Q_k,
$$
where $Q_k:=\{\xi:\xi\in\bR^d,\nu_j^k\xi_j\ge0,j\in\bN_d\}$. Suppose
that we have chosen the $k$-th $d$-hyperoctant and desire that for
each $f\in L^2(\bR^d)$ its multidimensional analytic signal $\cA f$
be such that
\begin{equation}\label{concentratek}
(\cA f)\hat{\,}(\xi)=\left\{\begin{array}{cc}2^d\hat{f}(\xi),&\xi\in Q_k,\\
0,&\mbox{otherwise.}\end{array}\right.
\end{equation}
The above equation has the following equivalent form
$$
(\cA
f)\hat{\,}(\xi)=\hat{f}(\xi)\prod_{j\in\bN_d}(1+\nu^k_j\sgn(\xi_j)),\
\ \xi\in\bR^d.
$$
By (\ref{multiplier2}), the above equations holds true for all $f\in
L^2(\bR^d)$ if and only if $\cA$ is defined by (\ref{analytic2})
with $T$ there being given as
$$
T=\prod_{j\in\bN_d}(H_0+i\nu^k_j H_j)-H_0,
$$
which satisfies (\ref{composition}).

\subsection{Computational Advantages}
As a consequence of (\ref{multiplier2}), it can be seen that a
bounded linear operator $T$ has the form (\ref{composition}) if and
only if there exists some function $m\in L^\infty(\bR^d)$ that is
constant on each $d$-hyperoctants such that
\begin{equation}\label{multiplier3}
(Tf)\hat{\,}=m\hat{f},\ \ f\in L^2(\bR^d).
\end{equation}
In other words, operators (\ref{composition}) are given by a very simple Fourier
multiplier. Therefore, efficient numerical algorithms based on fast
multidimensional discrete Fourier transforms (see \cite{BCK} and the
references cited therein) can be developed for the computation of
the multidimensional analytic signal (\ref{analytic2}).

\subsection{Physical Requirements}

Following \cite{Vakman1,Vakman2}, we ask what continuous linear
operators $T$ preserve the class of complex sinusoids. To make this
precise, we need to extend the definition of partial Hilbert
transforms. For each $\alpha\in \bZ_{d+1}^d$, we define $\cH_\alpha:
L^1(\bR^d)\to \cS'(\bR^d)$ by setting for each $f\in L^1(\bR^d)$ as
\begin{equation}\label{defl1}
\langle \cH_\alpha
f,\varphi\rangle:=\prod_{j\in\bN_d}(-1)^{\min(1,\alpha_j)}\
\int_{\bR^d}f\cH_\alpha\varphi dx,\ \ \varphi\in \cS(\bR^d),
\end{equation}
where for $\phi$ in a locally convex space $V$ and $u\in V^*$
$\langle u,\phi\rangle:=u(\phi)$. Since $\cH_\alpha$ is bounded from
$\cS(\bR^d)$ to $L^\infty(\bR^d)$, (\ref{defl1}) is well-defined.
Denote by $\bH^1(\bR^d)$ the Banach space of functions $f\in
L^1(\bR^d)$ such that $H_jf\in L^1(\bR^d)$ endowed with the norm
$$
\|f\|_{\bH^1(\bR^d)}:=\sum_{j\in\bZ_{d+1}}\|H_jf\|_{L^1(\bR^d)}.
$$
The dual of $\bH^1(\bR^d)$ is the space $\mbox{BMO}(\bR^d)$ of
locally integrable functions on $\bR^d$ with bounded mean
oscillation \cite{Stein}, which embeds continuously into
$\cS'(\bR^d)$. Note that
$L^\infty(\bR^d)\subseteq\mbox{BMO}(\bR^d)$. It is clear that
$\cH_\alpha$ is bounded from $\bH^1(\bR^d)$ to $\bH^1(\bR^d)$ for
each $\alpha\in\bZ_{d+1}^d$. Thus, we define for
$f\in\mbox{BMO}(\bR^d)$ $\cH_\alpha f$ again as a BMO function by
\begin{equation}\label{defbmo}
\langle \cH_\alpha
f,g\rangle:=\prod_{j\in\bN_d}(-1)^{\min(1,\alpha_j)}\langle
f,\cH_\alpha g\rangle,\ \ g\in \bH^1(\bR^d).
\end{equation}

For $\xi\in\bR^d$, we set $\sgn(\xi):=(\sgn(\xi_j):j\in\bR^d)$ and
shall characterize operators that preserve the class of complex
sinusoids.

\begin{prop} Suppose that $T$ is a linear operator from
$L^\infty(\bR^d)\cup L^2(\bR^d)$ to $\cS'(\bR^d)$ that is continuous
in the topology of $\cS'(\bR^d)$ and satisfies for some function
$\lambda:\bC^d\to\bC$ that
\begin{equation}\label{sgnlambda}
Te^{i(\omega,\cdot)}=\lambda(\sgn(\omega))e^{i(\omega,\cdot)},\ \
\mbox{for all }\xi\in\bR^d.
\end{equation}
Then $T$ has the form (\ref{composition}).
\end{prop}
\begin{proof}
Let $f\in L^2(\bR^d)$. Since
$\bS:=\span\{e^{i(\omega,\cdot)}:\omega\in\bR^d\}$ is dense in
$\cS'(\bR^d)$, there exists a sequence $f_n\in\bS$ that converges to
$f$ in $\cS'(\bR^d)$. By the continuity of $T$, $Tf_n$ converges to
$Tf$. Thus, $(Tf_n)\hat{\,}\to (Tf)\hat{\,}$. Denote for
$\omega\in\bR^d$ by $\delta_\omega$ the delta distribution at
$\omega$. We observe that
$$
(e^{i(\omega,\cdot)})\hat{\,}=(2\pi)^d\delta_\omega.
$$
Consequently,
\begin{equation}\label{deltafourier}
(Te^{i(\omega,\cdot)})\hat{\,}=(2\pi)^d\lambda(\sgn(\omega))\delta_\omega.
\end{equation}
Let $\varphi$ be an arbitrary function in
$\cS_k(\bR^d):=\{\varphi\in\cS(\bR^d):\supp\varphi\subseteq Q_k\}$
for some $k\in\bN_{2^d}$. Then we see that
$\lambda(\sgn(\cdot))\varphi\in\cS(\bR^d)$.  Assume that
$f_n=\sum_{j\in\bN_{k_n}}c_{nj}e^{i(\omega_{nj},\cdot)}$. We get by
(\ref{deltafourier}) that
$$
\langle
(Tf_n)\hat{\,},\varphi\rangle=\sum_{j\in\bN_{k_n}}c_{nj}(2\pi)^d\lambda(\sgn(\omega_{nj})\varphi(\omega_{nj})=\langle
\hat{f_n},\lambda(\sgn(\cdot))\varphi\rangle.
$$
Since
$$
\lim_{n\to\infty}\langle (Tf_n)\hat{\,},\varphi\rangle=\langle
(Tf)\hat{\,},\varphi\rangle\mbox{ and
}\lim_{n\to\infty}\langle\hat{f_n},\lambda(\sgn(\cdot))\varphi\rangle=\langle
\hat{f},\lambda(\sgn(\cdot))\varphi\rangle,
$$
we have for all $\varphi\in\cS_k(\bR^d)$, $k\in\bN_{2^d}$ that
$$
\langle (Tf)\hat{\,},\varphi\rangle=\langle
\lambda(\sgn(\cdot))\hat{f},\varphi\rangle.
$$
Since the linear span of $\cup_{k\in\bN_{2^d}}S_k(\bR^d)$ is dense in
$\cS(\bR^d)$, the above equation remains true for all $\varphi\in
\cS(\bR^d)$. Therefore, $(Tf)\hat{\,}=\lambda(\sgn(\cdot))\hat{f}$.
We complete the proof by pointing out that a linear operator
$T:L^2(\bR^d)\to L^2(\bR^d)$ has the form (\ref{composition}) if and
only if it is defined by a Fourier multiplier that is constant on
each $Q_k$, $k\in\bN_{2^d}$.
\end{proof}

\subsection{The Bedrosian Identity}

Considering the importance of the Bedrosian identity in the
time-frequency of one-dimensional signals, it is natural for us to
expect a linear operator $T$ in the definition of the
multidimensional analytic signal (\ref{analytic2}) to satisfy a
multidimensional Bedrosian identity. Assume that a bounded linear
operator $T: L^2(\bR^d)\to L^2(\bR^d)$ is engaged to define
multidimensional analytic signal by (\ref{analytic2}). First of all,
we shall require that $T$ is translation invariant, namely, for all
$x\in\bR^d$ and $f\in L^2(\bR^d)$
$$
T(f(\cdot-x))=(Tf)(\cdot-x).
$$
The reason is that there is usually a delay in time of an input in
signal processing. Thus we would like the processor $T$ to have the
property that a delay in the input will simply cause the same delay
in the output.

Now we are seeking extensions of the Bedrosian theorem that a
bounded linear translation invariant operator could possibly
satisfy. We say that a bounded linear operator $T:L^2(\bR^d)\to
L^2(\bR^d)$ satisfies the {\it type one Bedrosian theorem} if the
Bedrosian identity
\begin{equation}\label{bedrosian2}
T(fg)=fTg
\end{equation}
holds true for all $f,g\in\cS(\bR^d)$ that satisfy
\begin{equation}\label{cube}
\supp\hat{f}\subseteq\prod_{j\in\bN_d}[-a_j,a_j],\
\supp\hat{g}\subseteq\prod_{j\in\bN_d}\bR\setminus(-a_j,a_j),\
\mbox{for some }a=(a_j\ge0:j\in\bN_d).
\end{equation}

The following results was proved in \cite{VZ}.

\begin{lemma}\label{lemmavz}
A bounded linear translation invariant operator $T:L^2(\bR^d)\to
L^2(\bR^d)$ satisfies the type one Bedrosian theorem if and only if
it is of the form (\ref{composition}).
\end{lemma}

Thus, if we desire that the operator $T$ satisfy the type one
Bedrosian theorem then it must be a linear combination of the
compositions of the partial Hilbert transforms. A more natural extension
of the Bedrosian theorem is the {\it type two Bedrosian theorem}, that is,
(\ref{bedrosian2}) holds true for all $f,g\in\cS(\bR^d)$ satisfying
\begin{equation}\label{ball}
\supp\hat{f}\subseteq B(0,R),\ \supp\hat{g}\subseteq \bR^d\setminus
B^o(0,R),\ \  \mbox{for some }R\ge0,
\end{equation}
where $B(x_0,r):=\{x\in\bR^d:\|x-x_0\|\le r\}$ for $x_0\in\bR^d$ and
$r>0$, and $B^o(x_0,r)$ denotes its interior. When $d=1$, the type one and type two Bedrosian theorems are
identically the same. However, we shall prove that for $d\ge2$ there
exist no nontrivial operators that satisfy the type two Bedrosian
theorem.

\begin{theorem} Let $d\ge 2$. Then there does not exist a nontrivial bounded linear translation invariant operator $T:L^2(\bR^d)\to
L^2(\bR^d)$ that satisfies the type two Bedrosian theorem.
\end{theorem}
\begin{proof} Let $T:L^2(\bR^d)\to
L^2(\bR^d)$ be bounded linear and translation invariant. Assume that
it satisfies the type two Bedrosian identity. Let $f,g$ be arbitrary
functions in $\cS(\bR^d)$ that satisfies (\ref{cube}). Let
$R:=\|a\|$. Then we observe that (\ref{ball}) holds true. By the
assumption, the Bedrosian identity (\ref{bedrosian2}) holds. We
conclude that (\ref{bedrosian2}) holds for all $f,g\in\cS(\bR^d)$
satisfying (\ref{cube}) for some $a:=(a_j\ge0:j\in\bN_d)$. By Lemma
\ref{lemmavz}, $T$ must be of the form (\ref{composition}). Thus,
$T$ is defined by (\ref{multiplier3}) via a Fourier multiplier $m\in
L^\infty(\bR^d)$ that is constant in each $d$-hyperoctant. Let $m_k$
be the constant value that $m$ takes in $Q_k$, $k\in\bN_{2^d}$.

Assume that $T$ is neither the zero operator nor a multiple of the
identity operator. In other words, $m_{k_1}\ne m_{k_2}$ for some
$k_1,k_2\in\bN_{2^d}$. Since $d\ge 2$, there must also exist some
$k_3\in\bN_{2^d}$ such that $m_{k_1}\ne m_{k_3}$ or $m_{k_2}\ne
m_{k_3}$. There hence exists a pair $l_1,l_2\in\bN_{2^d}$ such that
$m_{l_1}\ne m_{l_2}$ and $\nu^{l_1}_{j_0}\nu^{l_2}_{j_0}=1$ for some
$j_0\in\bN_d$. For notational simplicity, assume that $l_1=1$ and
$j_0=1$. Since $l_2\ne 1$, the set $\{j\in\bN_d:\nu^{l_2}_j=-1\}$ is
nonempty. As a consequence, there exist positive constants
$\epsilon,r_1,r_2$ so that $\xi^1,\xi^2\in\bR^d$ defined by
$\xi^1:=(1-\frac\epsilon2,r_1,r_1,\ldots,r_1)$ and
$$
\xi^2_k:=\left\{\begin{array}{cc} 1-\epsilon,&k=1,\\
\frac{r_1}2,& \nu^{l_2}_k=1,\\
-r_2,&\nu^{l_2}_k=-1.
\end{array}\right.
$$
satisfy $\|\xi^1\|<1$, $\|\xi^2\|>1$ and $\|\xi^1-\xi^2\|<1$. Choose
$r>0$ so small that $B(\xi^1-\xi^2,r)\subseteq Q_1\cap B(0,1)$ and
$B(\xi^2,r)\subseteq Q_{l_2}\cap(\bR^d\setminus B(0,1))$. One can
construct nonnegative $f,g\in\cS(\bR^d)$ such that
$\supp\hat{f}\subseteq B(\xi^1-\xi^2,r)$, $\supp\hat{g}\subseteq
B(\xi^2,r)$, $\hat{f}(\xi^1-\xi^2)>0$, and $\hat{g}(\xi^2)>0$. Then
we have that
$$
\supp\hat{f}\subseteq B(0,1),\ \ \supp\hat{g}\subseteq
\bR^d\setminus B^o(0,1).
$$
Since $T$ satisfies the type two Bedrosian theorem,
(\ref{bedrosian2}) holds. Applying the Fourier transform to both
sides of (\ref{bedrosian2}) yields that
$$
\int_{\bR^d}\hat{f}(\xi-\eta)\hat{g}(\eta)(m(\xi)-m(\eta))d\eta=0,\
\ \mbox{for all }\xi\in\bR^d,
$$
which has the following form at $\xi=\xi^1$,
$$
(m_1-m_{l_2})\int_{B(\xi^2,r)}\hat{f}(\xi^1-\eta)\hat{g}(\eta)d\eta=0.
$$
However, since $m_1\ne m_{l_2}$,
$\hat{f}(\xi^1-\xi^2)\hat{g}(\xi^2)>0$, and
$\hat{f}(\xi^1-\cdot)\hat{g}$ is nonnegative,
$$
(m_1-m_{l_2})\int_{B(\xi^2,r)}\hat{f}(\xi^1-\eta)\hat{g}(\eta)d\eta\ne0.
$$
We hence come to a contradiction and complete the proof.
\end{proof}

Based on the above mathematical considerations, we conclude that the
only operators that should be used in (\ref{analytic2}) to define
the multidimensional analytic signal should be linear combinations
of the compositions of the partial Hilbert transforms.

\section{Multidimensional Bedrosian Identities}
\setcounter{equation}{0}

The purpose of the this section is to study the Bedrosian identity
(\ref{bedrosian2}) for the time-frequency analysis of
multidimensional analytic signals. By the discussion in the last
section, we shall assume from now on that $T$ is of the form
(\ref{composition}). In other words, there exist constants
$m_k\in\bC$, $k\in\bN_{2^d}$ such that $T$ is defined by
(\ref{multiplier3}) through the Fourier multiplier $m\in
L^\infty(\bR^d)$ given as
\begin{equation}\label{tmultiplier}
m(\xi):=m_k,\ \ \xi\in Q_k,\ k\in\bN_{2^d}.
\end{equation}
Denote by $T^*$ the adjoint operator of $T$ on $L^2(\bR^d)$, that
is,
$$
(Tf,g)_{L^2(\bR^d)}=(f,T^*g)_{L^2(\bR^d)},\ \ f,g\in L^2(\bR^d),
$$
where $(\cdot,\cdot)_{L^2(\bR^d)}$ is the inner product on
$L^2(\bR^d)$. It is clear that $T^*$ is defined by the Fourier
multiplier $\overline{m}$. By definition (\ref{defl1}), as a linear
combination of compositions of partial Hilbert transform, $T$ is
well-defined on $L^1(\bR^d)$ with range in $\cS'(\bR^d)$ by
\begin{equation}\label{definetl1}
\langle Tf,\varphi\rangle=\langle
f,\overline{T^*\overline{\varphi}}\rangle, \ \ \varphi\in
\cS(\bR^d).
\end{equation}

We shall start the characterization of the Bedrosian identity with a
lemma on a property of the operator $T$. To this end, we recall that
the Fourier transform $\hat{S}$ of a temperate distribution
$S\in\cS'(\bR^d)$ is again a temperate distribution satisfying
\begin{equation}\label{fouriercs}
\langle \hat{S},\varphi\rangle=\langle S, \hat{\varphi}\rangle,\ \
\varphi\in\cS(\bR^d).
\end{equation}

\begin{lemma}\label{fouriert}
For all $f\in L^1(\bR^d)$, $(Tf)\hat{\,}=m\hat{f}$.
\end{lemma}
\begin{proof}
Let $f\in L^1(\bR^d)$. We proceed by equation (\ref{fouriercs}) and
(\ref{definetl1}) that for all $\varphi\in\cS(\bR^d)$
\begin{equation}\label{fourierteq1}
\langle \hat{Tf},\varphi\rangle=\langle Tf,
\hat{\varphi}\rangle=\langle
f,\overline{T^*\overline{\hat{\varphi}}}\rangle.
\end{equation}
Set $h:=(\overline{T^*\overline{\hat{\varphi}}})\check{\,}$, the
inverse Fourier transform of
$\overline{T^*\overline{\hat{\varphi}}}$. Note that $h\in
L^1(\bR^d)$. Thus, by (\ref{fourierteq1})
$$
\langle
(Tf)\hat{\,},\varphi\rangle=\int_{\bR^d}f\hat{h}dx=\int_{\bR^d}\hat{f}hdx..
$$
We compute that
$$
h=(\overline{T^*\overline{\hat{\varphi}}})\check{\,}=\frac1{(2\pi)^d}\overline{(T^*\overline{\hat{\varphi}})\hat{\,}}=
\overline{(T^*(\overline{\varphi})\check{\,})\hat{\,}}=\overline{\overline{m}\,\overline{\varphi}}=m\varphi.
$$
Combining the above two equations proves the lemma.
\end{proof}

\begin{theorem}\label{characterization}
Let $T$ be defined by the Fourier multiplier (\ref{tmultiplier}).
Then $f,g\in L^2(\bR^d)$ satisfies the Bedrosian identity
(\ref{bedrosian2}) if and only if
\begin{equation}\label{characeq}
\sum_{j\in\bN_{2^d},\, j\ne
k}(m_k-m_j)\int_{Q_j}\hat{f}(\xi-\eta)\hat{g}(\eta)d\eta=0,\ \
\xi\in Q_k,\ k\in\bN_{2^d}.
\end{equation}
\end{theorem}
\begin{proof}
Since the Fourier transform is injective from $\cS'(\bR^d)$ to
$\cS'(\bR^d)$, (\ref{bedrosian2}) holds if and only if
$$
(T(fg))\hat{\,}=(fTg)\hat{\,}.
$$
By Lemma \ref{fouriert} and writing the Fourier transform of a
product of two functions as their convolution, the above equations
is equivalent to
\begin{equation}\label{characeq1}
m(\xi)\int_{\bR^d}\hat{f}(\xi-\eta)g(\eta)d\eta=\int_{\bR^d}\hat{f}(\xi-\eta)m(\eta)\hat{g}(\eta)d\eta,\
\ \ae \xi\in\bR^d.
\end{equation}
Since the left hand side above is continuous about $\xi\in\bR^d$,
(\ref{characeq1}) holds if and only if
\begin{equation}\label{characeq2}
m(\xi)\int_{\bR^d}\hat{f}(\xi-\eta)g(\eta)d\eta=\int_{\bR^d}\hat{f}(\xi-\eta)m(\eta)\hat{g}(\eta)d\eta,\
\ \mbox{for all } \xi\in\bR^d.
\end{equation}
By (\ref{tmultiplier}), (\ref{characeq2}) has the form
(\ref{characeq}). The proof is complete.
\end{proof}

In the one-dimensional case, various characterizations of the
Bedrosian identity have been proposed in the literature. For
instance, the above theorem for $d=1$ was proved in \cite{YZ2}. A
similar result was presented in \cite{Br2} under the assumptions
that $f,g\in L^2(\bR)\cap L^\infty(\bR)$. Other characterizations
can be found in \cite{XY,YZ1}.

We now turn to sufficient conditions for the Bedrosian identity. We
first prove the following generalization of the Bedrosian theorem,
which was first proved in \cite{Stark} for the case when $T$ is the
{\it total Hilbert transform} $\prod_{j\in\bN_d}H_j$.

\begin{prop}\label{bedrosiantheorem}
If $f,g\in L^2(\bR^d)$ satisfy for some $a:=(a_j\ge0:j\in\bN_d)$ and
$b:=(b_j\ge0:j\in\bN_d)$ that
\begin{equation}\label{bedrosiantheoremwq}
\supp\hat{f}\subseteq\prod_{j\in\bN_d}[-a_j,b_j],\
\supp\hat{g}\subseteq\prod_{j\in\bN_d}\bR\setminus(-b_j,a_j)
\end{equation}
then the Bedrosian identity (\ref{bedrosian2}) holds true.
\end{prop}
\begin{proof} It suffices to prove that the integrand in each
integral of (\ref{characeq}) vanishes almost everywhere. Let $\xi\in
Q_k$ and $j\ne k\in \bN_{2^d}$. We need to show that
$\hat{f}(\xi-\eta)\hat{g}(\eta)$ is zero for almost every $\eta\in
Q_j$. Assume that there exists $\eta\in Q_j$ that is in the support
of $\hat{f}(\xi-\cdot)\hat{g}$. Then $\eta\in \supp\hat{g}\cap Q_j$
and $\xi-\eta\in \supp\hat{f}$. Since $j\ne k$, there exists
$l\in\bN_d$ such that $\xi_l\eta_l\le 0$. We may assume that
$\xi_l\ge0$ and $\eta_l\le 0$. By (\ref{bedrosiantheoremwq}),
$\eta_l<-b_l$ while $\xi_l-\eta_l\le b_l$, which is impossible. The
contradiction completes the proof.
\end{proof}

To present the next sufficient condition, we set for
$j\in\bN_{2^d}$, $I_j:=\{k\in\bN_{2^d}:m_k=m_j\}$ and
$\cQ_j:=\cup_{k\in I_j}Q_k$. We also call a subset $A\subseteq\bR^d$
closed under addition if for all $x,y\in A$, $x+y\in A$.

\begin{prop}\label{halfsupport}
Let $j\in\bN_{2^d}$. If $\supp\hat{f}\cup\supp\hat{g}\subseteq \cQ_j
$ and $\cQ_j$ is closed under addition then $f,g\in L^2(\bR^d)$
satisfy the Bedrosian identity (\ref{bedrosian2}).
\end{prop}
\begin{proof} By Theorem \ref{bedrosiantheoremwq}, it suffices to
show that for all $\xi\in\bR^d$
\begin{equation}\label{halfsupporteq1}
\int_{\cQ_j}\hat{f}(\xi-\eta)\hat{g}(\eta)(m(\xi)-m_j)d\eta=0.
\end{equation}
Clearly, the above equation holds true for $\xi\in\cQ_j$. Assume
that there exists $\xi\in\bR^d\setminus\cQ_j$ that does not satisfy
(\ref{halfsupporteq1}). Thus, there must exist $\eta\in\cQ_j$ such
that $\xi-\eta\in\supp\hat{f}\subseteq\cQ_j$. Since $\cQ_j$ is closed under
addition, $\xi-\eta\in\cQ_j$ and $\eta\in\cQ_j$ imply that
$\xi\in\cQ_j$, a contradiction.
\end{proof}

In the one-dimensional case, the Bedrosian theorem has an appealing
physical interpretation. Namely, it states that if $f\in L^2(\bR)$
has low Fourier frequency and $g\in L^2(\bR)$ has high Fourier
frequency then they satisfy the Bedrosian identity
(\ref{bedrosian1}). Similarly, condition (\ref{bedrosiantheoremwq}) can be interpreted as that $f$ has low Fourier frequency in each
coordinate while $g$ is of high Fourier frequency in each
coordinate. It had been conjectured that the condition in the
Bedrosian theorem was necessary for the one-dimensional Bedrosian
identity (\ref{bedrosian1}) until an explicit contradicting example
was constructed in \cite{YZ2}. We recall that
$$
f(t):=\frac{1}{1+t^2},\ g(t):=\frac{1-2t^2}{4+5t^2+t^4},\ \ t\in\bR
$$
satisfy $H(fg)=fHg$ while $\supp\hat{f}=\supp\hat{g}=\bR$. Using
this example, we define
$$
F(x):=\prod_{j\in\bN_d}f(x_j),\ \ G(x):=\prod_{j\in\bN_d}g(x_j),\ \
x\in\bR^d.
$$
Since $T$ is of the form (\ref{composition}), $T(FG)=F(TG)$. This
together with $\supp\hat{F}=\supp\hat{G}=\bR^d$ implies that the
condition (\ref{bedrosiantheoremwq}) is also unnecessary for the
multidimensional Bedrosian identity (\ref{bedrosian2}).

Surprisingly, a necessity about the one-dimensional Bedrosian
theorem was obtained in \cite{YZ1}. It asserts that if $f,g\in
L^2(\bR)$ satisfy the Bedrosian identity (\ref{bedrosian1}), $\supp
\hat{f}\subseteq [a, b]$ for some $a,b\ge0$, and endpoints $a,b$ are
indeed contained in $\supp \hat{f}$ then $\supp\hat{g}\subseteq
\bR\setminus (-b,a)$. We shall present an extension of this result
for the multidimensional identity (\ref{bedrosian2}). To this end,
we recall that the convolution $\varphi*\psi$ of $\varphi,\psi\in
L^2(\bR^d)$ is defined by
$$
(\varphi*\psi)(x):=\int_{\bR^d}\varphi(x-t)\psi(t)dt,\ \ x\in\bR^d.
$$
It is well-known that $\supp \varphi*\psi\subseteq\supp
\varphi+\supp \psi$. In the case that both $\varphi,\psi$ are
compactly supported, we have the celebrated Titchmarsh convolution
theorem \cite{Lions,T}.

For each subset $A\subseteq\bR^d$, we denote by $\co A$ the convex
hull of $A$ in $\bR^d$.

\begin{lemma}\label{titchmarsh}
If both $\varphi,\psi\in L^2(\bR^d)$ are
compactly supported then
$$\co \supp(\varphi*\psi)=\co \supp \varphi+\co\supp \psi.$$
\end{lemma}

\begin{theorem}\label{necessity} Let $\nu$ be an extreme point of
$[-1,1]^d$, $Q:=\{\xi\in\bR^d:\xi_j\nu_j\ge0,\ j\in\bN_d\}$,
$Q':=-Q$, $m|_{Q}\ne m|_{Q'}$, and $f,g\in L^2(\bR^d)$ satisfy
$\supp\hat{f},\supp\hat{g}\subseteq Q\cup Q'$. If
$a:=(a_j>0:j\in\bN_d)$ and $b:=(b_j>0:j\in\bN_d)$ that
$(a_j\nu_j:j\in\bN_d),\ (-b_j\nu_j:j\in\bN_d)\in\supp\hat{f}$ and
$$\supp\hat{f}\subseteq\{\xi\in Q:\xi_j\nu_j\le a_j,\
j\in\bN_d\}\cup\{\xi\in -Q:-\xi_j\nu_j\le b_j,\ j\in\bN_d\},$$ then
$g\in L^2(\bR^d)$ satisfies the Bedrosian identity
(\ref{bedrosian2}) if and only if
$$
\supp\hat{g}\subseteq\biggl\{\xi\in
Q:\sum_{j\in\bN_d}\frac{\nu_j\xi_j}{b_j}\ge
1\biggr\}\bigcup\biggl\{\xi\in
-Q:\sum_{j\in\bN_d}\frac{-\nu_j\xi_j}{a_j}\ge 1\biggr\}.
$$
\end{theorem}

Before moving on to the proof, let us understand the conditions in
Theorem \ref{necessity}. Suppose that $\nu=\nu^1$. Consequently,
$Q=Q_1$ and $Q'=Q_2$. Theorem \ref{necessity} states that given
$f,g\in L^2(\bR^d)$ with $\supp\hat{g}\subseteq Q_1\cup Q_2$, if
$$
\supp\hat{f}\subseteq\prod_{j\in\bN_d}[0,a_j]\bigcup\prod_{j\in\bN_d}[-b_j,0],\mbox{
for some }a=(a_j>0:j\in\bN_d),\ b=(b_j>0:j\in\bN_d)
$$
and $a,-b$ are actually contained in $\supp\hat{f}$ then $g$
satisfies the Bedrosian identity if and only if any
$\xi\in\supp\hat{g}$ lies in
$$
\biggl\{\xi:\xi_j\ge0,\sum_{j\in\bN_d}\frac{\xi_j}{b_j}\ge
1\biggr\}\mbox{ or
}\biggl\{\xi:\xi_j\le0,\sum_{j\in\bN_d}\frac{\xi_j}{a_j}\le
-1\biggr\}.
$$
We illustrate the supports of $\hat{f}$ and $\hat{g}$ for the
two-dimensional case in the following graph.

\begin{center}
\scalebox{0.9}[0.6]{\includegraphics*{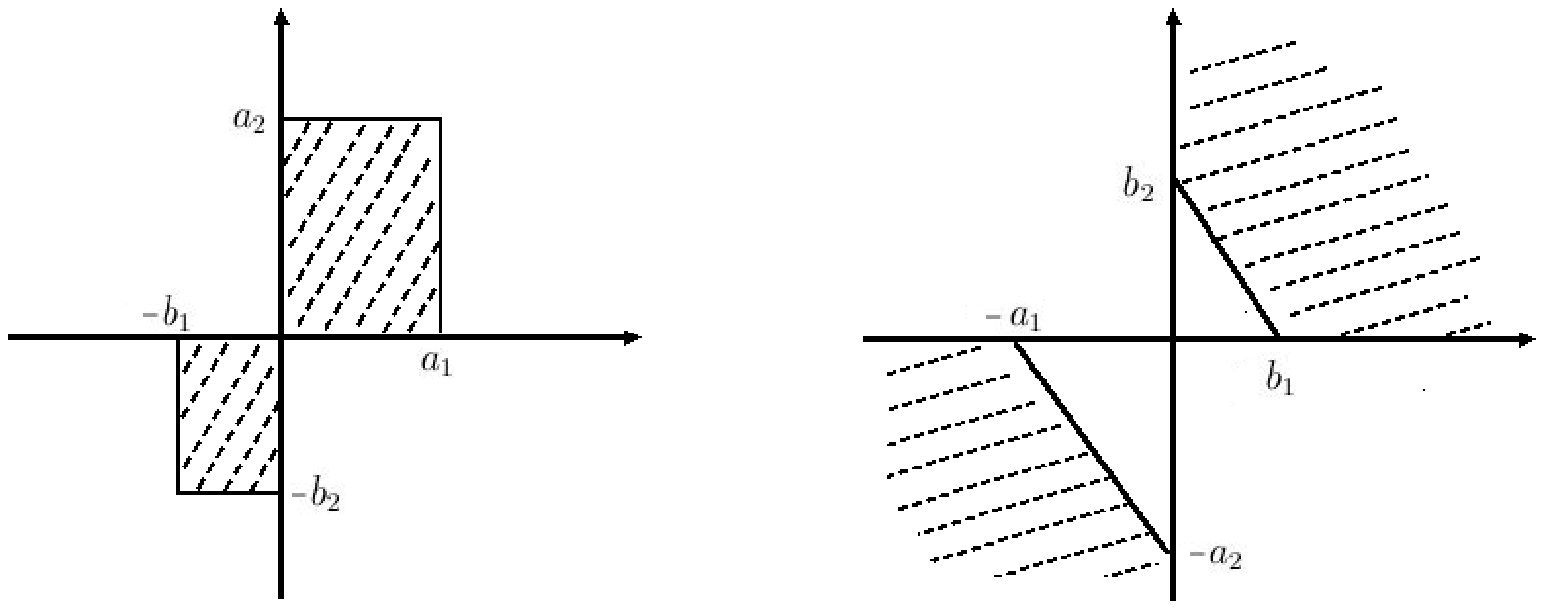}}
\end{center}
\begin{center}
{\footnotesize  {\bf  (a)}
$\supp\hat{f}$\quad\quad\quad\quad\quad\quad\quad\quad\quad\quad\quad\quad\quad\quad\quad\quad\quad\quad\quad\quad\quad
{\bf (b)} $\supp\hat{g}$}
\end{center}

We next present the proof of Theorem \ref{necessity}.

\begin{proof} Without loss of generality, assume that $Q=Q_1$ and
$Q'=Q_2$. By the assumptions that
$\supp\hat{f},\supp\hat{g}\subseteq Q_1\cup Q_2$ and $m_1\ne m_2$,
we obtain from Theorem \ref{characterization} that $g$ satisfies the
Bedrosian identity (\ref{bedrosian2}) if and only if
\begin{equation}\label{necessityeq1}
\int_{Q_2}\hat{f}(\xi-\eta)\hat{g}(\eta)d\eta=0,\ \ \xi\in Q_1
\end{equation}
and
$$
\int_{Q_1}\hat{f}(\xi-\eta)\hat{g}(\eta)d\eta=0,\ \ \xi\in Q_2.
$$
We shall only prove that $\supp(\hat{g}\chi_{Q_2})\subseteq \{\xi\in
Q_2:\sum_{j\in\bN_d}\frac{\xi_j}{a_j}\le -1\}$ since the other
inclusion that $\supp(\hat{g}\chi_{Q_1})\subseteq \{\xi\in
Q_1:\sum_{j\in\bN_d}\frac{\xi_j}{b_j}\ge 1\}$ can be handled in a
similar way. Set $\varphi:=\hat{f}\chi_{Q_1}$ and
$\psi:=\hat{g}\chi_{Q_2}$. We decompose $\psi$ uniquely into
$$
\psi=\psi_1+\psi_2
$$
where $\supp\psi_1\subseteq\prod_{j\in\bN_d}[-a_j,0]$ and
$\supp\psi_2\subseteq Q_2\setminus\prod_{j\in\bN_d}(-a_j,0)$. By
equation (\ref{necessityeq1}),
$$
(\varphi*\psi_1)(\xi)=-(\varphi*\psi_2)(\xi),\ \xi\in Q_1.
$$
Since
$$
\supp(\varphi*\psi_2)\subseteq\supp\varphi+\supp\psi_2\subseteq
Q_2\setminus\prod_{j\in\bN_d}[-a_j,0]+\prod_{j\in\bN_d}[0,a_j]\subseteq
\bR^d\setminus Q_1^o,
$$
where $Q_1^o$ is the interior of $Q_1$, we have that
$$
\supp(\varphi*\psi_1)\subseteq \bR^d\setminus Q_1^o.
$$
This together with
$$
\supp(\varphi*\psi_1)\subseteq\supp\varphi+\supp\psi\subseteq\prod_{j\in\bN_d}[-a_j,a_j]
$$
implies that
\begin{equation}\label{necessityeq2}
\co(\varphi*\psi_1)\subseteq\{\xi:-a_j\le\xi_j\le a_j,\ j\in\bN_d,\
\sum_{j\in\bN_d}\frac{\xi_j}{a_j}\le d-1\}.
\end{equation}
We claim that for any $\xi\in\supp\psi_1$,
$\sum_{j\in\bN_d}\frac{\xi_j}{a_j}\le -1$. Assume to the contrary
that this is invalid. Then there exists $\xi^0\in \co\supp\psi_1$
such that $\sum_{j\in\bN_d}\frac{\xi^0_j}{a_j}> -1$. Note that
$\xi^1:=a\in\supp\varphi$, which satisfies
$\sum_{j\in\bN_d}\frac{\xi^1_j}{a_j}=d$. By Lemma \ref{titchmarsh},
$$
\co\supp(\varphi*\psi_1)=\co\supp\psi_1+\co\supp\varphi.
$$
It follows that there exists some point
$\xi\in\co\supp(\varphi*\psi_1)$ such that
$$
\sum_{j\in\bN_d}\frac{\xi_j}{a_j}>d-1,
$$
contradicting (\ref{necessityeq2}). We conclude that
$$
\supp\psi_1=\supp(\hat{g}\chi_{Q_2})\subseteq\{\xi\in
Q_2:\sum_{j\in\bN_d}\frac{\xi_j}{a_j}\le -1\},
$$
which completes the proof.
\end{proof}

By a direct application of the Titchmarsh convolution theorem, one obtains another necessity theorem analog to the one-dimensional one proved in \cite{YZ1}.

\begin{theorem}\label{essnecesity}
Let $f\in L^2(\bR^d)$ satisfy for  some $a:=(a_j\ge0:j\in\bN_d)$ and
$b:=(b_j\ge0:j\in\bN_d)$ that
$$
\supp\hat{f}\subseteq\prod_{j\in\bN_d}[-a_j,b_j]
$$
and all the extreme points of $\prod_{j\in\bN_d}[-a_j,b_j]$ are contained in $\supp\hat{f}$. Then $g\in\L^2(\bR^d)$ with $\hat{g}$ being compactly supported satisfies the Bedrosian identity $T(fg)=fTg$ for all operators $T$ of the form (\ref{composition}) if and only if
\begin{equation}\label{essnecesitycond2}
\supp\hat{g}\subseteq\prod_{j\in\bN_d}\bR\setminus(-b_j,a_j).
\end{equation}
\end{theorem}
\begin{proof}
The sufficiency has been proved in Proposition \ref{bedrosiantheorem}. One sees that $g$ satisfies $T(fg)=fTg$ for all operators $T$ given by (\ref{composition}) if and only if for all $j\in\bN_d$
\begin{equation}\label{essnecesityeq1}
H_j(fg)=fH_j(g).
\end{equation}
Set $\bR^d_{j+}:=\{\xi\in\bR^d:\xi_j\ge0\}$ and $\bR^d_{j-}:=\{\xi\in\bR^d:\xi_{j}\le0\}$. Introduce for each $\varphi\in L^2(\bR^d)$ the associated pair of functions defined by
$$
(\varphi_{j+})\hat{\,}=\left\{
\begin{array}{ll}
\hat{\varphi}(\xi),&\xi\in\bR^d_{j+},\\
0,&\mbox{elsewhere},
\end{array}
\right.\quad
(\varphi_{j-})\hat{\,}=\left\{
\begin{array}{ll}
\hat{\varphi}(\xi),&\xi\in\bR^d_{j-},\\
0,&\mbox{elsewhere},
\end{array}
\right.
$$
By Theorem \ref{characterization}, identity (\ref{essnecesityeq1}) implies
$$
\supp(\hat{f_{j+}}*\hat{g_{j-}})\subseteq R^d_{j-},\ \ \supp(\hat{f_{j-}}*\hat{g_{j+}})\subseteq R^d_{j+}.
$$
By Lemma \ref{titchmarsh}, we get
$$
\supp(\hat{f_{j+}})+\supp(\hat{g_{j-}})\subseteq R^d_{j-},\ \ \supp(\hat{f_{j-}})+\supp(\hat{g_{j+}})\subseteq R^d_{j+}.
$$
Since the exterme points of $\prod_{j\in\bN_d}[-a_j,b_j]$ lie in $\supp\hat{f}$, there exist $\xi\in\supp\hat{f_{j+}}$ and $\eta\in\hat{f_{j-}}$ with $\xi_j=b_j$ and $\eta_j=-a_j$. This together with the above equation imply
$$
\supp\hat{g}\subseteq\{\xi\in\bR^d:\xi_j\in\bR\setminus(-b_j,a_j).
$$
Since this is true for each $j\in\bN_d$, we get (\ref{essnecesitycond2}).
\end{proof}
\section{Basic Multidimensional Analytic Signals}
\setcounter{equation}{0}

Assume that an operator $T$ of the form (\ref{composition}) has been
chosen to define multidimensional analytic signals. Applying $T$
directly to a given signal would generally not yield physically
meaningful results for the reason that the signal may contain
multiple components. It is desirable to first decompose the signal
into a sum of basic signals that behave well under $T$, and then
apply the operator $T$ to each summand. Note that $T$ given by
(\ref{composition}) is a linear combination of compositions of
partial Hilbert transforms. Thus, the purpose of this section is to
construct basic signals that will behave well under each partial
Hilbert transforms. More precisely, we shall first construct
one-dimensional signals in
$$
\cM:=\{\rho\cos\theta: \rho\in L^2(\bR),\ \theta\in C^1(\bR),\
\rho\ge0,\ \theta'\ge0,\ H(\rho\cos\theta)=\rho\sin\theta\}.
$$
Multidimensional basic signals can then be formed by tensor products
of functions in $\cM$.

Constructions of functions $\rho\cos\theta$ in $\cM$ has recently
been considered in \cite{QWXZ}, where the phase function $\theta$
was selected to be a strictly increasing function on $\bR$.
Consequently, the constructed signal $\rho\cos\theta$ does not
possess much fluctuation, which, on the other hand, is usually
required in the time-frequency analysis. Thus, we shall adopt a
different approach by choosing periodic phase functions that still
enjoy nonnegative derivative. Specifically, we shall use phase
functions determined by a finite Blaschke product.

Set $\bU:=\{z\in\bC:|z|<1\}$ and let $\ba:=(a_j:j\in\bN_n)\in
\bU^n$. The finite Blaschke product $\cB_\ba$ associated with $\ba$
is a holomorphic function on $\bU$ defined as
$$
\cB_\ba(z):=\prod_{j\in\bN_n}\frac{z-a_j}{1-\overline{a_j}z},\ \
z\in\bU.
$$
We let $\theta_\ba$ be the phase function determined by $\cB_\ba$ as
$$
e^{i\theta_\ba(t)}=\cB_\ba(e^{it}), \ t\in\bR,
$$
and aim at characterizing all real functions $f\in L^2(\bR)$ such
that
\begin{equation}\label{ftheta}
H(f\cos\theta_\ba)=f \sin\theta_\ba.
\end{equation}

We start with a straightforward observation.

\begin{lemma}\label{equivalence}
A real function $f\in L^2(\bR)$ satisfies (\ref{ftheta}) if and only
if
\begin{equation}\label{singleeq1}
H(fe^{i\theta_\ba})=-i fe^{i\theta_\ba}.
\end{equation}
\end{lemma}
\begin{proof}
Let $f\in L^2(\bR)$ be real. Assume that (\ref{ftheta}) holds true
then we apply $H$ to both sides of (\ref{ftheta}) to get that
$$
-f\cos\theta_\ba=H(f\sin\theta_\ba).
$$
Combining the above equation with (\ref{ftheta}) yields that
$$
H(fe^{i\theta_\ba})=H(f\cos\theta_\ba)+iH(f\sin\theta_\ba)=f\sin\theta_\ba-if\cos\theta_a=-i
fe^{i\theta_\ba},
$$
which is (\ref{singleeq1}). Conversely, if (\ref{singleeq1}) holds
true then by comparing the real part of its both sides, we obtain
(\ref{ftheta}). Thus, it suffices to show that (\ref{singleeq1}) is
equivalent to (\ref{singleeq}).
\end{proof}

We first deal with the situation when $\ba$ is a singleton $\{a\}$,
$a\in\bU$. In this case, we abbreviate $\theta_\ba$ as $\theta_a$.
Denote by $\tau$ the backshift operator defined by
$$
(\tau f)(t)=f(t-1),\ \ t\in\bR.
$$

\begin{lemma}\label{single}
If $\ba$ consists of a single point $a$ for some $a\in\bU$ then a
function $f\in L^2(\bR)$ satisfies (\ref{singleeq1}) if and only if
\begin{equation}\label{singleeq}
(\tau \hat{f})(\xi)=a\hat{f}(\xi),\mbox{ for all }\xi\le 0.
\end{equation}
\end{lemma}
\begin{proof} Let $f\in L^2(\bR)$. We compute that
\begin{equation}\label{singleeq2}
(fe^{i\theta_a})\hat{\,}=-a\hat{f}+(1-|a|^2)\sum_{k\in\bN}{\bar{a}}^{k-1}\tau^k\hat{f}.
\end{equation}
By taking the Fourier transform of both sides of (\ref{singleeq1})
and engaging (\ref{hilbertsgn}), we see that (\ref{singleeq1}) is
equivalent to
$$
(fe^{i\theta_a})\hat{\,}(\xi)=0,\ \ \xi\le 0.
$$
By (\ref{singleeq2}), the above equation can be rewritten as
\begin{equation}\label{singleeq3}
-a\hat{f}(\xi)+(1-|a|^2)\sum_{k\in\bN}{\bar{a}}^{k-1}(\tau^k\hat{f})(\xi)=0,\
\ \xi\le0.
\end{equation}
It remains to prove that (\ref{singleeq3}) and (\ref{singleeq})
imply each other. Assume that (\ref{singleeq}) holds true. Then we
observe for all $\xi\le 0$ that
$$
\begin{array}{rl}
\displaystyle{-a\hat{f}(\xi)+(1-|a|^2)\sum_{k\in\bN}{\bar{a}}^{k-1}\tau^k\hat{f}(\xi)}&\displaystyle{=-a\hat{f}(\xi)+(1-|a|^2)\sum_{k\in\bN}{\bar{a}}^{k-1}
a^k\hat{f}(\xi)}\\
&\displaystyle{=a\hat{f}(\xi)\left(-1+(1-|a|^2)\sum_{k\in\bN}(|a|^2)^{k-1}\right)=0}.
\end{array}
$$
On the other hand, assume that (\ref{singleeq3}) holds. We apply
$\tau$ to the left hand side of (\ref{singleeq3}) and then multiply
it by $\bar{a}$ to obtain that
$$
-|a|^2(\tau\hat{f})(\xi)+(1-|a|^2)\sum_{k=2}^\infty{\bar{a}}^{k-1}(\tau^k\hat{f})(\xi)=0,\
\ \xi\le0.
$$
Subtracting the above equation from (\ref{singleeq3}) gives
$$
-a\hat{f}(\xi)+|a|^2(\tau\hat{f})(\xi)+(1-|a|^2)(\tau\hat{f})(\xi)=0,\
\ \xi\le0.
$$
An rearrangement of the above equation yields (\ref{singleeq}) and
completes the proof.
\end{proof}

We next present the crucial lemma leading to the characterization of
real $f$ satisfying (\ref{ftheta}). To this end, we denote for
$\ba:=(a_j:j\in\bN_n)\in\bU^n$ by $\tau_\ba$ the operator
$$
(\tau_\ba f)=\biggl(\prod_{j\in\bN_n}(\tau-a_j)\biggr)f.
$$

\begin{lemma}\label{multiple} A function $f\in L^2(\bR)$ satisfies
(\ref{singleeq1}) if and only if
\begin{equation}\label{multipleeq}
(\tau_\ba\hat{f})(\xi)=0,\ \ \mbox{ for all }\xi\le0.
\end{equation}
\end{lemma}
\begin{proof} We shall use the induction on the number $n$ of
factors in the Blaschke product $\cB_\ba$. By Lemma \ref{single},
the result holds true when $n=1$. Assume that $n\ge2$ and the result
also holds for the $n-1$ case. Let $\ba':=(a_j:1\le j\le n-1)\in
\bU^{n-1}$ and $g=fe^{i\theta_{a_{n}}}$. Then (\ref{singleeq1}) is
equivalent to
$$
H(ge^{i\theta_{\ba'}})=-ige^{i\theta_{\ba'}}.
$$
By induction, the above equation holds true if and only if
$$
(\tau_{\ba'}\hat{g})(\xi)=0,\ \ \mbox{ for all }\xi\le 0.
$$
We have computed that
$$
\hat{g}=-a_n\hat{f}+(1-|a_n|^2)\sum_{k\in\bN}{\overline{a_n}}^{k-1}\tau^k\hat{f}.
$$
Therefore, it suffices to show that (\ref{multipleeq}) is equivalent
to
$$
-a_n\tau_{\ba'}\hat{f}(\xi)+(1-|a_n|^2)\tau_{\ba'}\sum_{k\in\bN}{\overline{a_n}}^{k-1}\tau^k\hat{f}(\xi)=0,\
\ \mbox{ for all }\xi\le0.
$$
Similar arguments as those in the proof of Lemma \ref{single}
fulfils this purpose.
\end{proof}

We conclude the results of this section by far into the following
theorem. Note that $\iota:=\tau^{-1}$ is the shift operator defined
as
$$
(\iota f)(t)=f(t+1),\ \ t\in\bR.
$$
We introduce another operator $\iota_{\bar{\ba}}$ by setting
$$
\iota_{\bar{\ba}}=\prod_{j\in\bN_n}(\iota-\overline{a_j}).
$$

\begin{theorem}\label{characf}
A real function $f\in L^2(\bR)$ satisfies (\ref{ftheta}) if and only
if
\begin{equation}\label{characfeq}
(\tau_\ba\hat{f})(\xi)=0,\ \mbox{for all }\xi\le0\mbox{ and
}(\iota_{\bar{\ba}}\hat{f})(\xi)=0,\ \mbox{for all }\xi\ge0.
\end{equation}
\end{theorem}
\begin{proof}
The result follow immediately from Lemmas \ref{equivalence},
\ref{single}, \ref{multiple}, and the fact that for a real function
$f\in L^2(\bR)$, $\overline{\hat{f}(\xi)}=\hat{f}(-\xi)$,
$\xi\in\bR$.
\end{proof}

We next seek a closed form for real functions $f$ satisfying
(\ref{ftheta}).
\begin{theorem}\label{explicit}
Let $\lambda_j$ be the distinct elements in $\ba$ with multiplicity
$n_j$, $1\le j\le k$. Then a real function $f\in L^2(\bR)$ satisfies
(\ref{ftheta}) if and only if there exist functions $c_{jl}\in
L^2(\bR)$ with $\supp(c_{jl})\hat{\,}\subseteq[-1,0]$, $1\le j\le
k$, $1\le l\le n_j$ such that
\begin{equation}\label{closedform}
f(x)=\Re\left(\sum_{j=1}^k\sum_{l=1}^{n_j}\frac{c_{jl}(x)}{(1-\lambda_j
e^{-ix})^l}\right).
\end{equation}
\end{theorem}
\begin{proof} Assume that $f$ is a real function in $L^2(\bR)$ that
satisfies (\ref{ftheta}). Then by Theorem \ref{characf}, $\hat{f}$
satisfies (\ref{characfeq}). It implies that for each
$\xi\in(-1,0]$, $\hat{f}(\xi-m)$ satisfies the difference equation
\begin{equation}\label{difference}
(\tau_\ba\hat{f})(\xi-m)=0,\ \ m\in\bZ_+.
\end{equation}
By the general solution of a difference equation, there exist
functions $e_{jl}$ on $(-1,0]$, $j\in\bN_k$, $l\in\bN_{n_j}$ such
that
\begin{equation}\label{generalsolution}
\hat{f}(\xi-m)=\sum_{j\in\bN_k}\sum_{l\in\bN_{n_j}}e_{jl}(\xi)\lambda_j^m\prod_{q=1}^{l-1}(m+q),\
\ m\in\bZ_+.
\end{equation}
Since $f$ is real, we have that

\begin{equation}\label{closedformeq1}
\begin{array}{rl}
f(x)&\displaystyle{=\frac1{\sqrt{2\pi}}\int_\bR
\hat{f}(\xi)e^{ix\xi}d\xi=\frac1\pi\Re\biggl(\int_{-\infty}^0\hat{f}(\xi)e^{ix\xi}d\xi\biggr)}\\
&\displaystyle{=\frac1\pi\Re\biggl(\sum_{m=0}^\infty\int_{-m-1}^{-m}\hat{f}(\xi)e^{ix\xi}d\xi\biggr)}\\
&\displaystyle{=\frac1\pi\Re\biggl(\sum_{m=0}^\infty\int_{-1}^{0}\hat{f}(\xi-m)e^{ix(\xi-m)}d\xi\biggr)}
\end{array}
\end{equation}
Now we apply (\ref{generalsolution}) to obtain that
\begin{equation}\label{closedformeq2}
\sum_{m=0}^\infty\int_{-1}^{0}\hat{f}(\xi-m)e^{ix(\xi-m)}d\xi=
\sum_{j\in\bN_k}\sum_{l\in\bN_{n_j}}\int_{-1}^0\biggl(\sum_{m=0}^\infty\lambda_j^me^{-imx}\prod_{q=1}^{l-1}(m+q)\biggr)e_{jl}(\xi)e^{ix\xi}d\xi.
\end{equation}
Recall that
\begin{equation}\label{closedformeq3}
\sum_{m=0}^\infty\lambda_j^me^{-imx}\prod_{q=1}^{l-1}(m+q)=\frac{(l-1)!}{(1-\lambda_je^{-ix})^l},\
\ j\in\bN_k,\ l\in\bN_{n_j}.
\end{equation}
Combining (\ref{closedformeq1}), (\ref{closedformeq2}), and
(\ref{closedformeq3}) yields that
$$
f(x)=\Re\biggl(\sum_{j\in\bN_k}\sum_{l\in\bN_{n_j}}\frac{1}{(1-\lambda_je^{-ix})^l}\frac{(l-1)!}\pi\int_{-1}^0e_{jl}(\xi)e^{ix\xi}d\xi\biggr).
$$
Finally, setting
$$
c_{jl}(x)=\frac{(l-1)!}\pi\int_{-1}^0e_{jl}(\xi)e^{ix\xi}d\xi,\ \
j\in\bN_k,\ l\in\bN_{n_j}
$$
leads to the form (\ref{characfeq}) of $f$.

Conversely, if $f$ os of the form (\ref{characfeq}) then by
reversing the above arguments, we obtain (\ref{generalsolution}).
Thus, for each $\xi\in(-1,0]$, $\hat{f}(\xi-m)$ satisfies the
difference equation (\ref{difference}). Therefore, $f$ satisfies
(\ref{multipleeq}). Since $f$ is real, (\ref{characfeq}) holds true.
By Lemma \ref{characf}, $f$ satisfies (\ref{ftheta}).
\end{proof}

{\small
\bibliographystyle{amsplain}
}
\end{document}